\documentclass[11pt, letterpaper]{article}
\usepackage[colorlinks=true,linkcolor=red,citecolor=blue]{hyperref}
\usepackage{fullpage,paralist,amssymb}
\usepackage{theorem}
\usepackage{booktabs}
\usepackage{stmaryrd}

\usepackage{mathptmx}

\usepackage[text={6.5in,9in}]{geometry}

\usepackage[compact]{titlesec}
\titlespacing{\section}{0pt}{3ex}{2ex}
\titlespacing{\subsection}{0pt}{2ex}{1ex}
\titlespacing{\subsubsection}{0pt}{0.5ex}{0ex}

\usepackage{verbatim}
\usepackage{amssymb}
\usepackage{amsmath}
\usepackage{ifthen}

\newtheorem{theorem}{Theorem}[section]

\newskip\theorempreskipamount
\newskip\theorempostskipamount
\global
\global

\newenvironment{reminder}[1]{\smallskip
\noindent {\bf Reminder of #1  }\em}{}
\newenvironment{proof}{\noindent {\bf Proof.  }}{\hfill$\Box$}
\newenvironment{proofof}[1]{\smallskip
\noindent {\bf Proof of #1.  }}{\hfill$\Box$}

\newtheorem{definition}{Definition}[section]
\newtheorem{lemma}{Lemma}[section]

\def \Z {{\mathbb Z}}
\def \R {{\mathbb R}}
\def \AC {{\sf AC}}

\def \poly {{\text{poly}}}
\def \P {{\sf P}}
\def \MAJ {{\sf MAJ}}
\def \SYM {{\sf SYM}}
\def \XOR {{\sf XOR}}
\def \AND {{\sf AND}}
\def \OR {{\sf OR}}
\def \THR {{\sf THR}}
\def \ACC {{\sf ACC}}

\def \NP {{\sf NP}}
\def \N {{\mathbb N}}
\def \F {{\mathbb F}}
\def \NEXP {{\sf NEXP}}

\def \eps {{\varepsilon}}

\title{New algorithms and lower bounds for circuits\\ with linear threshold gates}
\author{Ryan Williams\thanks{Supported by an Alfred P. Sloan Fellowship, a Microsoft Research Faculty Fellowship, a David Morgenthaler II Faculty Fellowship, and NSF CCF-1212372. Any opinions, findings, and conclusions or recommendations
expressed in this material are those of the author(s) and do not necessarily reflect the views of the National Science Foundation.}\\Stanford University}

\begin{document}
\maketitle

\begin{abstract} Let $\ACC \circ \THR$ be the class of constant-depth circuits comprised of  AND, OR, and MOD$m$ gates (for some constant $m > 1$), with a bottom layer of gates computing arbitrary linear threshold functions. This  class of circuits can be seen as a ``midpoint'' between $\ACC$ (where we know nontrivial lower bounds) and depth-two linear threshold circuits (where nontrivial lower bounds remain open).

We give an algorithm for evaluating an arbitrary symmetric function of $2^{n^{o(1)}}$ $\ACC \circ \THR$ circuits of size $2^{n^{o(1)}}$, on all possible inputs, in $2^n \cdot \poly(n)$ time. Several consequences are derived:

\begin{itemize}

\item The number of satisfying assignments to an $\ACC\circ\THR$ circuit of subexponential size can be computed in $2^{n-n^{\eps}}$ time (where $\eps > 0$ depends on the depth and modulus of the circuit). 

\item $\NEXP$ does not have quasi-polynomial size $\ACC \circ \THR$ circuits, and $\NEXP$ does not have quasi-polynomial size $\ACC \circ \SYM$ circuits. Nontrivial size lower bounds were not known even for ${\sf AND} \circ {\sf OR} \circ \THR$ circuits.

\item Every 0-1 integer linear program with $n$ Boolean variables and $s$ linear constraints is solvable in $2^{n-\Omega(n/((\log M)(\log s)^{5}))}\cdot \poly(s,n,M)$ time with high probability, where $M$ upper bounds the bit complexity of the coefficients. (For example, 0-1 integer programs with weights in $[-2^{\poly(n)},2^{\poly(n)}]$ and $\poly(n)$ constraints can be solved in $2^{n-\Omega(n/\log^6 n)}$ time.) Impagliazzo, Paturi, and Schneider~\cite{IPS13} recently gave an algorithm for $\tilde{O}(n)$ constraints; ours is the first asymptotic improvement over exhaustive search for for up to subexponentially many constraints.

\end{itemize}

We also present an algorithm for evaluating depth-two linear threshold circuits (a.k.a., $\THR \circ \THR$) with exponential weights and $2^{n/24}$ size on all $2^n$ input assignments, running in $2^n \cdot \poly(n)$ time. This is evidence that non-uniform lower bounds for $\THR \circ \THR$ are within reach. 

\end{abstract}

\thispagestyle{empty}
\newpage
\setcounter{page}{1}

\section{Introduction}

Recall that in the non-uniform Boolean circuit model, one designs an infinite family of logical circuits $\{C_n\}$, one for each input length $n$, in order to recognize a given binary language $L \subseteq \{0,1\}^{\star}$. This model is notoriously powerful, even when the size of $C_n$ is bounded from above by a fixed polynomial in $n$, defining the complexity class $\P/\poly$. With polynomial size circuits, one can already ``compute'' some undecidable languages, such as $L'=\{1^n~|~\text{the $n$th Turing machine halts on blank tape}\}$. 
Nevertheless, it is strongly believed that $\NP \not\subset \P/\poly$, meaning that for even modestly-sized instances of $\NP$-complete problems, the sizes of \emph{computations} on such instances must be inevitably gigantic. However, knowledge of $\P/\poly$ is rather poor, due to the ``infinite'' nature of the model: it is open if the huge complexity class \emph{nondeterministic exponential time} ($\NEXP$) is contained in $\P/\poly$. This containment would imply that problems verifiable with exponentially-long witnesses could be efficiently ``solved'' with small circuits. It looks obviously absurd; how can we rule it out?

In recent years, it has been demonstrated that the existence of nontrivial circuit-analysis algorithms is closely linked to the $\NEXP$ versus $\P/\poly$ problem. For instance, Impagliazzo, Kabanets, and Wigderson~\cite{IKW} showed that $\NEXP \not\subset \P/\poly$ follows, if there is a $2^{n^{o(1)}}$ time algorithm that can approximate a given circuit's acceptance probability to within $1/10$. They also proved a partial converse, in that $\NEXP \not\subset \P/\poly$ implies a certain kind of derandomization. Subsequent work~\cite{Williams10} strengthened the algorithms-to-lower bounds implication, proving that a similar algorithm which (for every $k$) runs in $2^{n-\omega(\log n)}$ time on all $n$-input $n^k$-size circuits still implies $\NEXP \not\subset \P/\poly$. A variant of this implication (for circuit satisfiability algorithms) was combined with an satisfiability algorithm for a restricted circuit class called $\ACC$, implying that $\NEXP$ does not have polynomial-size $\ACC$ circuits~\cite{Williams11}. Recently, it was shown that $\NEXP \not\subset \P/\poly$ is equivalent to establishing a ``weak'' form of natural proofs~\cite{WilliamsSTOC13}, building on Impagliazzo et al.\footnote{In particular, $\NEXP \not\subset \P/\poly$ if and only if there is a ``constructive'' property of Boolean functions that is ``useful'' against $\P/\poly$. The natural proofs barrier~\cite{RazborovRudich97} states that if such a property is also ``large'' (true of a large fraction of functions) then strong cryptographic pseudorandom generators do not exist. Hence, assuming strong crypto, $\NEXP$ lower bounds must somehow confront the framework of natural proofs but sidestep the ``large'' condition.}

To continue progress on circuit lower bounds for $\NEXP$, it is imperative to understand algorithms for analyzing circuits, such as algorithms for circuit satisfiability, evaluating a circuit on all $2^n$ inputs, and approximating the acceptance probability of a circuit.\footnote{Recent surveys on these issues include~\cite{Williams11survey,Santhanam12,Cohen13,Oliveira13}.} In this paper, we make this sort of algorithmic progress for circuits with arbitrary \emph{linear threshold} gates: such a gate outputs $1$ if and only if a certain linear inequality $\sum_i w_i x_i \geq t$ is true, where $w_i, t \in \Z$ are \emph{weights} and $x_i \in \{0,1\}$ are inputs to the gate. Linear threshold functions have been studied for decades, coinciding with research on neural networks~\cite{Minsky-Papert69,Muroga71}. Low-depth linear threshold circuits are powerful: many basic functions in arithmetic, algebra, and cryptography are known to be implementable with only \emph{constant-depth} linear threshold circuits~\cite{Reif-Tate92,Siu-Bruck93,Siu-R94,Maciel-Therien99,Naor-Reingold04}. In terms of lower bounds for such circuits, very weak questions remain major open problems: for example, is all of $\NEXP$ solvable with polynomial-size depth-\emph{two} linear threshold circuits with exponential-size weights?\footnote{Note that for thresholds with polynomially-bounded weights, depth-two lower bounds are known; however depth-three lower bounds are still open. The survey of Razborov~\cite{Razborov92} is still relatively current on these points.} Depth-two circuits correspond to \emph{multilayer perceptrons} with only one hidden layer. Despite considerable study in neural networks and deep learning, we still lack understanding of the power of depth-two.

In this paper, we report some new progress on understanding the power of linear threshold gates.

\paragraph{Algorithms and lower bounds for ACC with threshold gates} Let $\ACC \circ \THR$ denote the class of circuits consisting of AND, OR, MOD$m$ gates for some constant $m$,\footnote{A MOD$m$ gate outputs $1$ if and only if the sum of its input bits is divisible by $m$.} and linear threshold gates, with unbounded fan-in and constant depth, such that the inputs of all linear threshold gates connect directly to the circuit's input variables. Let $\SYM \circ \ACC \circ \THR$ be the class of circuits where the output gate computes an arbitrary symmetric function, and its inputs connect to the outputs of $\ACC\circ \THR$ circuits. We show that such circuits can very efficiently evaluated on all $2^n$ inputs, even if they are of $2^{n^{o(1)}}$ size.

\begin{theorem} \label{evalaccthr} Given a $\SYM \circ \ACC \circ \THR$ circuit with $n$ inputs and $2^{n^{o(1)}}$ size, we can produce its outputs on all $2^n$ inputs in $2^n \cdot \poly(n)$ time.

More generally, such a circuit of size $s$ can be evaluated on all inputs in $2^n \cdot \poly(\log s, n) + 2^{O(\log s)^c}$ time, for some $c \geq 1$ depending on the depth of the circuit and the modulus $m$ of its MOD$m$ gates.
\end{theorem}

The proof of Theorem~\ref{evalaccthr} also carries through for $\SYM \circ \ACC \circ \SYM$, where the bottom layer gates compute arbitrary symmetric functions (i.e., functions which only depend on the number of true inputs) of $2^{n^{o(1)}}$ wires. This algorithm can be used to \emph{count} the number of satisfying assignments to $\ACC\circ\THR$ circuits.

\begin{theorem}\label{COUNTaccthr} For every integer $m > 1$ and $d > 0$, there is an $\eps > 0$ such that counting satisfying assignments to $\ACC\circ\THR$ circuits of size $2^{n^{\eps}}$, depth $d$, and MOD$m$ gates can be done in $2^{n-n^{\eps}}$ time.
\end{theorem}

By modifying prior arguments~\cite{Williams11}, we can conclude lower bounds for such circuits. The new argument shows that the ability to count SAT assignments entails non-uniform lower bounds for circuit classes with very weak closure properties.

\begin{theorem} \label{nexpaccthr} $\NEXP$ does not have non-uniform $\ACC \circ \THR$ circuits of quasi-polynomial size. 
\end{theorem}

As Theorem~\ref{evalaccthr} also holds for $\SYM \circ \ACC \circ \SYM$, it follows that $\NEXP$ doesn't have $\ACC \circ \SYM$ circuits of quasi-polynomial size. as well

Twenty years ago, Maciel and Therien~\cite{Maciel-Therien93} considered lower bounds for $\AC^0 \circ \MAJ$ circuits (which $\ACC \circ \THR$ subsumes), but nontrivial lower bounds have not been reported. Regan~\cite{Regan97} studied ${\sf MOD}_2 \circ \AND \circ \THR$ circuits and also noted the absence of lower bounds. Lower bounds have been open even for the much weaker class $\AND \circ \OR \circ \MAJ$~\cite{Hansen-Podolskii13}. 

Theorem~\ref{nexpaccthr} moves a little closer to an ``unconditional break'' of the natural proofs barrier~\cite{RazborovRudich97}. That is, it seems plausible that pseudorandom functions can be implemented with $\ACC \circ \THR$ circuits, in which case any lower bounds proved against such circuits must be non-naturalizing.\footnote{It is not completely settled whether the proof that $\NEXP \not\subset \ACC$ is ``truly'' non-naturalizing; it could be that the natural proofs barrier is irrelevant to the problem. (If pseudorandom functions cannot be implemented in $\ACC$, then natural proofs considerations don't apply to $\ACC$ anyway; if such functions can be implemented in $\ACC$, then the $\NEXP$ lower bound is indeed non-naturalizing.)} Plaku~\cite{Plaku02} observed that the Naor-Reingold family of pseudorandom functions~\cite{Naor-Reingold04} can be implemented with quasi-polynomial size $\OR \circ \THR \circ \AND$ circuits; it follows that the natural proofs barrier already applies to this circuit class. It is an interesting open problem if $\ACC\circ\THR$ can efficiently simulate such depth-three circuits.

Building on Theorem~\ref{evalaccthr}, we also give a new method for solving 0-1 integer linear programs. In FOCS'13, Impagliazzo, Paturi, and Schneider~\cite{IPS13} showed that for each $c > 1$, there is a $\delta < 1$ such that 0-1 integer LPs with $cn$ constraints can be solved in $2^{\delta n}$ time. We provide an improvement over exhaustive search for up to subexponentially many constraints:

\begin{theorem} \label{IP} Every 0-1 integer linear program with $n$ variables and $s$ constraints can be solved in time $2^{n-\Omega(n/((\log M)(\log s)^{5}))}\cdot \poly(s,n,M)$ with high probability, where $M \leq 2^{o(n)}$ upper bounds the bit complexity of the coefficients in the program.
\end{theorem}

Notice that the theorem allows for enormous coefficients, of size up to $2^{2^{o(n)}}$. The time bound compares favorably with the $\AC^0$ circuit satisfiability bounds of Impagliazzo, Matthews, and Paturi~\cite{ImpagliazzoMP12}: there, the authors use random restriction methods to solve satisfiability of $\AC^0$ circuits with depth $d$ and size $s$ in $2^{n-n/(\log s)^{O(d)}}$ randomized time with zero error. Our algorithm shows that, using probabilistic polynomials and fast rectangular matrix multiplication, one can obtain similar running times for SAT of $\AC^0[2]$ circuits with a layer of symmetric gates at the bottom.

\paragraph{Depth-two linear threshold circuit evaluation.} We take an important step towards depth-two linear threshold circuit (a.k.a. $\THR \circ \THR$) lower bounds for the case of exponential weights, by giving an efficient algorithm for evaluating such circuits on all possible assignments.

\begin{theorem} \label{depth2eval} Let $k > 1$. Given a depth-two $2^{n/24}$-size linear threshold circuit $C$ with integer weights in $[-2^{n^k},-2^{n^k}]$, we can evaluate $C$ on all $2^n$ input assignments in $2^n \cdot \poly(n^k)$ time.
\end{theorem}

Theorem~\ref{depth2eval} follows from a more general result showing that any sufficiently large ``combinatorial rectangle'' of inputs can be evaluated in $\poly(n)$ amortized time per input. Noting that a similar statement for evaluating ACC circuits forms the heart of the proof of $\NEXP \not\subset \ACC$~\cite{Williams11}, Theorem~\ref{depth2eval} suggests that large complexity classes (such as $\NEXP$) cannot have small depth-two linear threshold circuits. However, we do not yet know how to turn Theorem~\ref{depth2eval} into depth-two linear threshold lower bounds.\footnote{The current theorems connecting circuit evaluation algorithms to circuit lower bounds require that, from the OR of a collection of circuits, we can generate an equivalent circuit in the same class. We do not know how to convert a large OR of $\THR \circ \THR$ circuits into an equivalent $\THR \circ \THR$ circuit, even assuming $\NEXP$ has small $\THR \circ \THR$ circuits. (In the case of ACC, this is trivial, because an OR of ACC circuits is still an ACC circuit.)}

\subsection{Prior work}

Considerable effort has been expended in proving lower bounds against circuits with linear threshold gates. Here we will provide some major highlights, in addition to the work already mentioned. 

It will help to introduce a little (standard) notation. Define  $\MAJ$, $\AND$, $\OR$, $\THR$, and $\SYM$ to be the class of one-gate circuits corresponding to MAJORITY, AND, OR, linear threshold, and symmetric functions, respectively, with ``free'' NOT gates that can appear after the output or on the input wires to the gate. (Recall that a symmetric Boolean function's output only depends on the number of true inputs.) For classes of circuits ${\cal C}$ and ${\cal D}$, define ${\cal C} \circ {\cal D}$ to be the class of circuits formed by taking a circuit $C \in {\cal C}$, and feeding the outputs of circuits from ${\cal D}$ as inputs to $C$. That is, ${\cal C} \circ {\cal D}$ is simply the composition of circuits from ${\cal C}$ and ${\cal D}$, with the circuits from ${\cal D}$ receiving the input and the circuit from ${\cal C}$ giving the output. We will equivocate the \emph{size} of a circuit with the number of wires, i.e., the number of directed arcs in the DAG defining the circuit. This is an important measure for circuits with symmetric gates, as the number of wires governs the size of the symmetric function representation.

Much work on depth-two threshold lower bounds has concentrated on lower bounds for inner product modulo $2$, i.e., $\text{IP2}(x_1,\ldots,x_n,y_1,\ldots,y_n) = \sum_i x_i \cdot y_i \bmod 2$. Note that IP2 is easy for $\ACC$ (being a MOD2 of AND gates). In groundbreaking work, Hajnal et al.~\cite{Hajnal93} proved that every $\MAJ \circ \MAJ$ circuit requires $2^{\Omega(n)}$ gates to compute IP2. They also showed $\MAJ \circ \SYM$ circuits can be efficiently simulated by $\MAJ \circ \MAJ$ circuits, so small $\MAJ \circ \SYM$ circuits also cannot compute IP2. Nisan~\cite{Nisan94} extended the lower bound to $\MAJ \circ \THR$ circuits, and Forster et al.~\cite{Forster01} extended the lower bound  to $\THR \circ \MAJ$ circuits. More recently, Sherstov~\cite{Sherstov09} showed that $\AC^0$  requires exponential-size $\MAJ \circ \MAJ$ circuits, Razborov and Sherstov~\cite{Razborov-Sherstov10} proved that depth-three $\AC^0$ requires exponential-size $\MAJ\circ \THR$ circuits, and Beame and Huynh~\cite{Beame-Huynh12} showed that $\AC^0$ requires $n^{\Omega(\log n)}$-size $\MAJ \circ \SYM \circ \AND$ circuits. 

Although superpolynomial-size lower bounds against $\MAJ \circ \AC^0$, $\THR \circ \AC^0$, $\MAJ \circ \MAJ \circ \AND$ and even $\MAJ \circ \MAJ \circ \AC^0$ circuits are known~\cite{Aspnes94,
Goldmann97,Razborov-Wigderson93,Hansen-Miltersen04}, and many lower bounds are known for $\AC^0$ circuits augmented with a small number of threshold gates~\cite{Beigel94,Barrington-Straubing94,Chat-Hansen05,Viola06,Hansen07,Gopalan-Servedio10,Lovett-Srinivasan11,Podolskii12}, lower bounds for $\AC^0 \circ \MAJ$ circuits have remained open. Maciel and Therien~\cite{Maciel-Therien93} conjectured that the majority-of-majority function is not in $\AC^0 \circ \MAJ$. 

Recently, Hansen and Podolskii~\cite{Hansen-Podolskii13} have shown an intriguing reduction: superpolynomial-size $\THR\circ\THR$ lower bounds for a function $f$ would follow from superlogarithmic lower bounds on the 3-party NOF unbounded-error communication complexity of $f$. 

\subsection{Comparison and Intuition}

It is instructive to discuss how this paper's approach relates to prior work on depth-two threshold lower bounds. A certain popular approach~\cite{Forster01,Lokam-book,Sherstov09,Razborov-Sherstov10} applies ingredients from Fourier analysis of Boolean functions, linear algebra,  communication complexity, discrepancy theory, \emph{etc.} In particular, these works follow the general scheme:
\smallskip

\begin{compactenum}
\item Define some notion of ``relaxed rank'' of a $2^{n/2} \times 2^{n/2}$ Boolean matrix $C$. Intuitively, if $C$ has ``relaxed rank'' $r$, then there are $2^{n/2} \times r$ and $r \times 2^{n/2}$ matrices $A$ and $B$ such that the entries of $A \cdot B$ correspond to the entries of $C$ in a direct way.
\item Show that every function $f : (\{0,1\}^{n/2} \times \{0,1\}^{n/2}) \rightarrow \{0,1\}$ computable with a ``small'' ${\cal C}$ circuit has ``small relaxed rank'' when construed as an $2^{n/2} \times 2^{n/2}$ Boolean matrix.
\item Show that some explicit family of functions $g_n : (\{0,1\}^{n/2} \times \{0,1\}^{n/2}) \rightarrow \{0,1\}$, construed as $2^{n/2} \times 2^{n/2}$ Boolean matrices, requires ``high relaxed rank'' asymptotically. 
\end{compactenum}
\smallskip

Together, these steps prove that the family $g := \{g_n\}$ cannot have ``small'' ${\cal C}$ circuits.

To prove $\ACC \circ \THR$ circuit lower bounds, we define a generalized rank notion we call the \emph{symmetric rank}, informally measuring how efficiently a $0$-$1$ matrix $M$ can be decomposed into a sum of rank-one matrices such that, after applying a fixed symmetric function to each entry of the sum, we obtain the matrix $M$. Combining several elements from previous work, we show that for a Boolean matrix representing the truth table of a $\SYM \circ \ACC \circ \THR$ circuit of size $s$, its symmetric rank is $O(2^{\log^c s})$ for some constant $c \geq 1$, depending on the depth $d$ and modulus $m$ of the MOD$m$ gates in the circuit. Moreover, given such a circuit we can efficiently compute a low-rank decomposition. 

However, we do not know how to use existing methods to prove that an explicit function $g$ has high symmetric rank. Instead, we take a more \emph{computational} approach that still exploits the low symmetric rank property. The idea is that, if we can efficiently compute a low-rank decomposition from a given circuit, then the circuit's truth table can be obtained faster than evaluating the circuit on all its inputs one-by-one. This in turn suggests that these circuits possess considerable structure that make them unsuitable for simulating very complex functions, such as those in $\NEXP$. 

Suppose we are given an $\SYM\circ \ACC\circ\THR$ circuit $C$ of size $s$ with $n$ inputs. Let $M$ be a $2^{n/2} \times 2^{n/2}$ matrix defining the function computed by $C$. First we show how given any such $C$ we can compute $2^{n/2} \times 2^{\log^c s}$ and $2^{\log^c s} \times 2^{n/2}$ matrices $A$ and $B$ (and a symmetric function $f$) giving a symmetric rank decomposition of $M$, in $2^{n/2} \cdot 2^{O(\log^c s)}$ time. By multiplying $A$ and $B$ and applying $f$ to each entry of the output matrix,  we can obtain $M$. When $s$ is sufficiently small, a rectangular matrix multiplication of Coppersmith~\cite{Coppersmith82} can be applied to compute the product of $A$ and $B$, and the final matrix $M$ is obtained in $\poly(n)$ time per entry. Hence, given an $\SYM \circ \ACC \circ \THR$ circuit $C$ of size $2^{n^{o(1)}}$, we can evaluate $C$ on all its $2^n$ inputs in only $2^n \cdot \poly(n)$ time. This fast evaluation algorithm is combined with prior work~\cite{Williams10,Williams11} along with some new tricks to exhibit a $g := \{g_n\} \in \NEXP$ which does not have quasipolynomial-size $\ACC \circ \THR$ circuits.

Our evaluation algorithm for depth-two threshold circuits (Theorem~\ref{depth2eval}) also uses Coppersmith's rectangular matrix multiplication as a subroutine, but the rest of the algorithm is rather different from the evaluation algorithm for $\SYM \circ \ACC \circ \THR$. We reduce the problem of efficiently evaluating a depth-two threshold circuit on many inputs to a special type of matrix multiplication. Namely, for two matrices $A$ and $B$ over the integers, we compute a ``weighted'' matrix product \[C[i,j] = \sum_k w_k \cdot \text{LEQ}(A[i,k],B[k,j]),\] where $\text{LEQ}(x,y)$ is a Boolean-valued function equal to $1$ if and only if $x \leq y$, and the $w_k$'s are arbitrary integer weights given as parameters to the problem. We show how Coppersmith's algorithm can be combined with a mild brute force search to efficiently compute a rectangular matrix product of the above form. 

\section{Algorithms and lower bounds for ACC with a layer of threshold gates}

The main theorem of this section is:

\begin{reminder}{Theorem~\ref{evalaccthr}} Given a $\SYM \circ \ACC \circ \THR$ circuit with $n$ inputs and $2^{n^{o(1)}}$ size, we can produce its outputs on all $2^n$ inputs in $2^n \cdot \poly(n)$ time. 

More generally, such a circuit of size $s$ can be evaluated on all inputs in $2^n \cdot \poly(\log s, n) + 2^{O(\log s)^c}$ time, for some $c \geq 1$ depending on the depth of the circuit and the modulus $m$ of its MOD$m$ gates.
\end{reminder}

\paragraph{Depth reduction.} The first stage of the proof is to convert an arbitrary $\SYM \circ \ACC \circ \THR$ circuit $C$ of size $s$ into a depth-two circuit $C''$ of symmetric gates, i.e., a $\SYM \circ \SYM$ circuit. The size of the depth-two circuit will be $O(2^{\log^c s})$ for a constant $c \geq 1$, depending on the (constant) depth and (constant) modulus of circuit $C$. This stage requires several different pieces from prior work.

\begin{lemma} \label{accthr2symsym} There is an algorithm which given an $\SYM \circ \ACC \circ \THR$ circuit $C$ of size $s \geq n$, depth $d$, and MOD$m$ gates, outputs an equivalent $\SYM \circ \SYM$ circuit $C''$ with at most $2^{(\log s)^c}$ wires, and runs in time $O(2^{(\log s)^c})$, for $c \geq 1$ depending only on $d$ and $m$.
\end{lemma}

The following paragraphs give the proof of Lemma~\ref{accthr2symsym}. Let $C$ be a $\SYM \circ \ACC \circ \THR$ circuit with inputs $x_1,\ldots,x_n$, size $s$, depth $d$, and MOD$m$ gates, for constants $d > 2$ and $m > 1$. In the proof, several constants arise; we will denote all of them by the same constant $b$ which is assumed to be the maximum of these quantities.

The first step in Lemma~\ref{accthr2symsym} is to translate the $\THR$ layer of $C$ into a $\SYM$ layer, by absorbing some of its complexity into the $\ACC$ part. Without loss of generality, we can assume that the weights of all threshold gates in $C$ have absolute value at most $2^{b n \log_2 n}$~\cite{Muroga-Toda-Takasu61,Muroga71}. (Every $\THR$ function is equivalent to one with weights of bit-complexity at most $b n \log_2 n$.)\footnote{In fact, this ``small-weight'' representation can be efficiently obtained, by evaluating the large-weight representation at only $n+1$ points, then solving a linear system in $n+1$ variables to determine the weights. See~\cite{Muroga-Toda-Takasu61}, Theorem 16.}

Maciel and Therien~\cite{Maciel-Therien98} provided several fairly tight low-deph circuits for various tasks. We need:

\begin{theorem}[\cite{Maciel-Therien98}, Theorem 3.3]\label{addition} Addition of $n$ distinct $n$-bit numbers can be performed with polynomial-size $\AND \circ \OR \circ \SYM$ circuits. Furthermore, the circuits can be constructed in polynomial time.
\end{theorem} 

We can therefore replace every $\THR$ gate of $C$ with an $\AC^0 \circ \MAJ$ circuit, as follows. Fix a threshold gate of $C$, with weights $w_{i_1},\ldots,w_{i_t}$ for $t \leq n$, computing $\sum_{j=1}^{t-1} w_{i_j}x_{i_j} \geq w_{i_t}$ for some $i_j \in \{1,\ldots,n\}$. Note $|w_{i_j}|\leq 2^{b n\log_2 n}$ for $j=1,\ldots,t$. Set $W = b n\log_2 n$. 

Let $D$ be a circuit for the addition of $t-1$ $W$-bit numbers, provided by Theorem~\ref{addition}. For $j=1,\ldots,t-1$, we connect to the $j$th $W$-bit input of $D$ a circuit which, given $x_{i_j}$, feeds $w_{i_j}$ to $D$ if the input bit $x_{i_j}=1$, and the all-zero $W$-bit string if $x_{i_j}=0$. Note this extra circuit actually contains no gates: it simply has a wire from $x_{i_j}$ to all bits of the $j$th $W$-bit input where the corresponding bit of $w_{i_j}$ equals $1$. Letting this new circuit be $D'$, we have $D'(x_1,\ldots,x_n) = \sum_{j=1}^{t-1} w_{i_j}x_{i_j}$. This can be compared to the value $w_{i_t}$ with an $\AC^0$ circuit, using the fact that the ``less-than-or-equal-to'' comparison of two integers can be performed in $\AC^0$~\cite{Chandra-Stockmeyer-Vishkin84}. We now have an $\AC^0 \circ \SYM$ circuit $D''$ of size $\poly(W,t) \leq n^b$ computing the given threshold gate. Applying this construction to each threshold gate in the $\THR$ layer of $C$, we obtain an $\SYM\circ \ACC\circ \SYM$ circuit $C'$ of size at most $s \cdot n^{b}$.

The next step of Lemma~\ref{accthr2symsym} is to convert the $\SYM \circ \ACC$ part into a $\SYM \circ \AND$ circuit, using a reduction of Beigel-Tarui~\cite{Beigel-Tarui} (with important details on constructibility filled in by Allender-Gore~\cite{Allender-Gore91}).

\begin{theorem}[\cite{Beigel-Tarui,Allender-Gore91}]\label{BT} Every $\SYM\circ \ACC$ circuit of size $s$ can be simulated by a $\SYM \circ \AND$ circuit of $2^{(\log s)^{c'}}$ size for some constant $c'$ depending only on the depth $d$ and MOD$m$ gates of the $\ACC$ part. Moreover, the $\AND$ gates of the final circuit have only $(\log s)^{c'}$ fan-in, the final circuit can be constructed from the original in $2^{O((\log s)^{c'})}$ time, and the final symmetric function at the output can be computed in $2^{O((\log s)^{c'})}$ time.
\end{theorem}

Applying this reduction to the top $\SYM \circ \ACC$ part of the circuit $C'$ results in an equivalent $\SYM \circ \AND_{(\log (s\cdot n^b))^{c'}} \circ \SYM$ circuit $C''$ of size $s' = 2^{O((\log (s \cdot n^b))^{c'})}$ (where the subscript on the $\AND$ denotes the fan-in of each AND gate). For simplicity of notation, let $t = (\log(s\cdot n^b))^{c'}$ in the following.

Extending a trick of Beigel~\cite{Beigel94} to symmetric gates, we can convert every $\AND_t \circ \SYM$ subcircuit of $C''$ with $n^b$ wires into a single $\SYM$ gate with $O(n^{b \cdot t})$ wires. Let $S_1(x_1,\ldots,x_n) \wedge \cdots \wedge S_t(x_1,\ldots,x_n)$ be one such subcircuit, where $S_i$ denotes the $i$th symmetric gate. In particular, for $i=1,\ldots,t$, let $f_i : \Z \rightarrow \{0,1\}$ be such that $f_i(\sum_{j=1}^n c_{i,j} x_j) = S_i(x_1,\ldots,x_n)$, where $c_{i,j}$ denotes the number of copies of $x_j$ that feed into $S_i$.

Let $B = 1+\max_{i} (\sum_{j=1}^n c_{i,j})$; note that $B \leq n^b$. Consider the linear form \[L(x_1,\ldots,x_n) = \sum_{i=1}^{t} B^{i-1} \cdot \left(\sum_{j=1}^n c_{i,j} x_j\right).\] For any Boolean assignment to the $x_j$'s, the number encoded by the linear form $L(x_1,\ldots,x_n)$ is an integer encoded in $O(t\cdot b \log n)$ bits. By construction, the bit representation of this integer contains, for every $i=1,\ldots,t$, the number of wires input to $S_i$ which are set true, as a string of $(b\log n)$ bits. Therefore, from the linear form $L(x_1,\ldots,x_n)$ we can easily infer whether all $S_i(x_1,\ldots,x_n)$ output $1$ or not, and hence output the value of $S_1 \wedge \cdots \wedge S_t$.

To implement this linear form with a single $\SYM$ gate, for all $j=1,\ldots,n$ we put $\sum_{i=1}^t B^{i-1} c_{i,j}$ wires from the input variable $x_j$ into the new $\SYM$ gate. Hence there are $O(n^{b \cdot t})$ wires from the inputs into this new $\SYM$ gate. By choosing the appropriate symmetric function (which outputs $1$ if and only if $L(x_1,\ldots,x_n)$ encodes a number such that $S_1 \wedge \cdots \wedge S_t$ is true) we can simulate any $\AND_t \circ \SYM$ circuit of $n^b$ wires with a single $\SYM$ gate of $O(n^{b \cdot t})$ wires.

Replacing each $\AND \circ \SYM$ subcircuit in this manner results in a $\SYM \circ \SYM$ circuit of size $O(s' \cdot n^{b \cdot t}) \leq 2^{O(\log s)^{c}}$ for some constant $c \geq 1$. This concludes the proof of Lemma~\ref{accthr2symsym}.

\paragraph{Symmetric rank.} Next, we prove that the truth table of any $\SYM \circ \SYM$ circuit $C''$ of $t$ wires and $n$ inputs represents a $2^{n/2} \times 2^{n/2}$ matrix of \emph{symmetric rank} at most $\poly(t)$, and this rank decomposition can be efficiently computed. For given matrices $A$ and $B$ over the integers, let $A \cdot B$ denote their matrix product over the integers. Let $M \in \{0,1\}^{m \times n}$. We define the \emph{symmetric rank of $M$} to be the minimum $r \in \N$ such that there are matrices $A \in \{0,1\}^{m \times r}$, $B \in \{0,1\}^{r \times n}$ and a function $f : \{0,1,\ldots,r\} \rightarrow \{0,1\}$ satisfying $M[i,j] = f((A \cdot B)[i,j])$ for all $i,j$. We call the triple $(A,B,f)$ a \emph{symmetric rank decomposition} of $M$. The symmetric rank is similar to the typical notion of rank, except for the additional function $f$ providing a ``filter'' from arbitrary integers back to $\{0,1\}$. This filter function could potentially lead to smaller rank decompositions than the typical notion. However, note the symmetric rank of $M$ is not necessarily at most (for instance) the rank of $M$ over $\R$, because $A$ and $B$ are required to have Boolean entries.

For simplicity let $n$ be even, and let $z_1,\ldots,z_{2^{n/2}}$ be the list of all $2^{n/2}$ $n/2$-bit strings in lexicographical order. For a circuit $C$ with $n$ inputs, define the \emph{truth table matrix} $M_C$ to be the $2^{n/2} \times 2^{n/2}$ matrix with $M_C[i,j]$ equal to the output of $C(z_i,z_j)$.

\begin{lemma}\label{symrank} Given a $\SYM\circ\SYM$ circuit $C$ with $t$ wires and $n$ inputs, its truth table matrix $M_C$ has symmetric rank $O(t^3)$, and a symmetric rank decomposition of $M_C$ can be computed from $C$ in $2^{n/2} \cdot \poly(t)$ time.
\end{lemma}

\begin{proof} For simplicity we assume $n$ is even; the case of odd $n$ will be apparent. Index the input variables of $C$ by $x_1,\ldots,x_n$. Let $g_1,\ldots,g_s$ be an indexing of the gates of $C$ on the bottom layer (closest to the inputs) and let $g'$ denote the output gate of $C$. (Note that $s \leq t$.) Let $f : \{0,1,\ldots,s\}\rightarrow\{0,1\}$ be the symmetric function of gate $g'$: for all $a \in \{0,1,\ldots,s\}$, $f(a) = b$ if and only if $a$ true inputs make $g'$ output $b$. 

We shall show how to efficiently construct matrices $A$ and $B$ with the appropriate properties. Let $z_1,\ldots,z_{2^{n/2}}$ be the list of all $n/2$-bit strings in lexicographical order, in the following. For every pair $(a,b) \in \{0,1,\ldots,t\}^2$ such that $a+b\leq t$, let $S_{a,b}\subseteq \{g_1,\ldots,g_s\}$ denote the subset of gates $g_j$ such that $a+b$ true inputs makes gate $g_j$ output $1$. 

The matrices $A$ and $B$ to be constructed show that the symmetric rank of $M_C$ is at most \[r = \sum_{a,b \in \{0,1,\ldots,t\} : a + b \leq t} |S_{a,b}| \leq O(t^3).\] In other words, each pair $(a,b)$ will add $|S_{a,b}|$ additional components to the rows of $A$ and the columns of $B$.

For $i=1,\ldots,2^{n/2}$, the $i$th row of $A$ and $i$th column of $B$ are defined as follows. For every pair $(a,b)$, allocate $|S_{a,b}|$ additional components for the rows of $A$ and columns of $B$. 

For $j=1,\ldots,|S_{a,b}|$, put a $1$ in the $j$th additional component of the $i$th row of $A$ if and only if there are $a$ true wires going into the $j$th gate of $S_{a,b}$ when the input variables $x_1,\ldots,x_{n/2}$ are given assignment $z_i$. That is, the $j$th component is $1$ if and only if the contribution (from the first half of variables) to the overall sum for the $j$th gate is $a$. 

Similarly, for $j=1,\ldots,|S_{a,b}|$, put a $1$ in the $j$th additional component of the $i$th column of $B$ if and only if there are $b$ true wires going into the $j$th gate of $S_{a,b}$ when the input variables $x_{n/2+1},\ldots,x_{n}$ are given assignment $z_i$. 

Note that each entry of $A$ and $B$ can be determined in $\poly(t)$ time. 

For every fixed $(a,b)$, the product of two $j$th components for the $i$th row of $A$ and the $k$th column of $B$ is either $0$ or $1$, and the product is $1$ if and only if:
\begin{compactitem}
\item the sum of true inputs into the $j$th gate of $S_{a,b}$ from the inputs $(x_1,\ldots,x_{n/2})$ equals $a$ when the inputs $(x_1,\ldots,x_{n/2})$ are assigned $z_i$,
\item the sum of true inputs into the same gate from $(x_{n/2+1},\ldots,x_n)$ equals $b$ when the inputs $(x_{n/2+1},\ldots,x_n)$ are assigned $z_k$, and 
\item the $j$th gate outputs $1$ when its sum of true inputs equals $a+b$.
\end{compactitem}
It follows that the \emph{inner product} of the $i$th row of $A$ and the $k$th column of $B$ equals the total number $N_{i,k}$ of true wires going into the output gate of $C$ on the variable assignment $(x_1,\ldots,x_n) \mapsto (z_i, z_k)$. By definition, $f(N_{i,k})$ equals the output of $C$ on that variable assignment.
\end{proof}

We need one more lemma to complete the proof of Theorem~\ref{evalaccthr}:

\begin{lemma}\label{coppersmith} For all sufficiently large $N$, and $\alpha \leq .172$, multiplication of an $N \times N^{\alpha}$ matrix with an $N^{\alpha} \times N$ matrix can be done in $N^2 \cdot \poly(\log N)$ arithmetic operations, over any field with $O(2^{\poly(\log N)})$ elements.\footnote{See Appendix~\ref{coppersmith-appendix} for an exposition of this result.}\end{lemma}

\begin{proofof}{Theorem~\ref{evalaccthr}} Given a $\SYM \circ \ACC \circ \THR$ circuit $C$ and size $s$, convert $C$ into a $\SYM \circ \SYM$ circuit $C''$ of $2^{(\log s)^c}$ size using Lemma~\ref{accthr2symsym}. Compute a symmetric rank decomposition of $C$ into $2^{n/2} \times 2^{3(\log s)^c}$ and $2^{3(\log s)^c} \times 2^{n/2}$ $0$-$1$ matrices $A$ and $B$ respectively, along with a function $f:[2^{3(\log s)^c}] \rightarrow \{0,1\}$. Compute the product of $A$ and $B$ in $2^n \cdot \poly(\log s, n)$ time, using Lemma~\ref{coppersmith}. Finally, evaluate function $f$ on all entries of the matrix product. This can be done by numerically sorting the entries, replacing each entry $v$ by $f(v)$, then inverting the sorted order, in time $2^n\cdot \poly(\log s, n) + 2^{O(\log s)^c}$. For $s \leq 2^{n^{o(1)}}$, the runtime is $2^n \cdot \poly(n)$. 
\end{proofof}

\subsection{Counting satisfying assignments to ACC of linear thresholds}

The evaluation algorithm of Theorem~\ref{evalaccthr} is quite powerful, substantially extending the class of circuits for which we can perform non-trivial circuit analysis. 

\begin{reminder}{Theorem~\ref{COUNTaccthr}} For every $m > 1$ and $d > 0$, there is an $\eps > 0$ such that counting satisfying assignments to $\ACC\circ\THR$ circuits of size $2^{n^{\eps}}$, depth $d$, and MOD$m$ gates can be done in $2^{n-n^{\eps}}$ time.
\end{reminder}

\begin{proof}  
For all $k \in \N$ and for $i=1,\ldots,2k$, define a $\text{Bit}^k_i$ function with $2^{2k}$ inputs as follows: for all $i=1,\ldots,2k$, $\text{Bit}^k_i$ outputs the $i$th bit of the sum of its input bits. Clearly, a $\text{Bit}^k_i$ function is symmetric.

Suppose we are given an $\ACC\circ\THR$ circuit $C$ of size $s$ and $n$ inputs, and we wish to count its satisfying assignments. Let $\ell < n/2$ be a parameter to set later. For every assignment $A_j \in \{0,1\}^{2\ell}$ to the last $2\ell$ inputs of $C$, make a copy of $C$ with the assignment $A_j$ plugged into those $2\ell$ inputs, calling this copy $C_{A_j}$. Note that each $C_{A_j}$ has (the same) $n-2\ell$ inputs $x_1,\ldots,x_{n-2\ell}$.

For every $i=1,\ldots,2\ell$, define $B_i(x_1,\ldots,x_{n-2\ell}) := \text{Bit}^{\ell}_i(C_{A_1}(x_1,\ldots,x_{n-2\ell}),\ldots,C_{A_{2^{2\ell}}}(x_1,\ldots,x_{n-2\ell}))$. Each function $B_i$ can be implemented in $s' = 2^{2\ell} \cdot s$ size, as a $\SYM \circ \ACC \circ \THR$ circuit. Applying Theorem~\ref{evalaccthr}, $B_i$ can be evaluated on all of its $2^{n-2\ell}$ possible assignments in time \[2^{n-2\ell}\cdot \poly(n) + 2^{\poly(\log s')} \leq 2^{n-2\ell}\cdot \poly(n) + 2^{\poly(\ell + \log s)}.\]

The above for-loop over all $i$ produces $2\ell \cdot 2^{n-2\ell}$ bits: for each of the $2^{n-2\ell}$ partial assignments to $n-2\ell$ variables, we learn the number (in $2\ell$ bits) of partial assignments on the other $2\ell$ variables which result in satisfaction. The number of all satisfying assignments is obtained by simply summing all $2\ell$-bit numbers obtained from the $2^{n-2\ell}$ assignments, in $2^{n-2\ell} \cdot \poly(\ell)$ time.

Letting $\ell = n^{\eps}/2$ for sufficiently small $\eps > 0$, we have a $2^{n-n^{\eps}}$ time algorithm.
\end{proof}

\subsection{Faster 0-1 linear programming}

$\ACC\circ\THR$ circuits are definitely powerful enough to simulate $0$-$1$ integer linear programming; a straightforward application of Theorem~\ref{COUNTaccthr} would yield a faster algorithm for the problem. However, the improvement over exhaustive search would be rather minor, and tedious to calculate.  
By modifying the proof of Theorem~\ref{evalaccthr} in appropriate places, we can derive a better algorithm in this case: 

\begin{reminder}{Theorem~\ref{IP}} Every 0-1 integer linear program with $n$ variables and $s$ constraints can be solved in time $2^{n-\Omega(n/((\log M)(\log s)^{5}))}\cdot \poly(s,n,M)$ with high probability, where $M \leq 2^{o(n)}$ upper bounds the bit complexity of the coefficients in the program.
\end{reminder}

\begin{proof} Consider a 0-1 linear program of the form $Ax \leq b$, along with a cost function $\langle c,x\rangle$ we wish to maximize, where $A \in \Z^{s \times n}$, $b \in \Z^s$, and $c \in ([-2^M,2^M] \cap \Z)^n$ by assumption on $M$. First, reduce the optimization problem to one of feasibility, in a standard way: include $\langle c, x \rangle \geq v$ as an additional constraint for various $v \in \Z$, and by binary searching on $v$, we maximize the value of $v$ such that the $s+1$ constraint system remains feasible. Since the $x_i$ are Boolean valued, the binary search uses at most $O(M + \log n)$ calls to feasibility questions.

Next, observe the feasibility questions can be viewed as a satisfiability question for a depth-two circuit $D$ with an AND at the top gate, and linear threshold gates on the bottom layer, by directly translating each constraint in the program into a linear threshold gate. By Theorem~\ref{addition} and the argument in Lemma~\ref{accthr2symsym}, each threshold gate in the circuit $D$ can be replaced with a polynomial-sized ${\sf LEQ} \circ \AND \circ \OR \circ \SYM$ circuit, where ${\sf LEQ}$ computes on $n$-bit integers $a$ and $b$ whether $a \leq b$. As ${\sf LEQ}$ has an $\OR\circ\AND\circ\sf XOR$ circuit of $O(n^2)$ size for $n$-bit inputs (see \cite{Chandra-Stockmeyer-Vishkin84} for a reference), the satisfiability question for the circuit $D$ reduces to the SAT question for an $\AC^0[2] \circ \SYM$ circuit $C$ where the $\AC^0[2]$ part has depth 5. Following the strategy of Theorem~\ref{COUNTaccthr} (and the author's ACC SAT algorithm~\cite{Williams11}), the satisfiability question for $C$ with $n$ inputs and size $\poly(s)$ can be efficiently converted into the problem of evaluating a larger $\AC^0[2] \circ \SYM$ circuit $C'$, where $C'$ has $n' = n-k$ inputs, $2^k \cdot \poly(s,M)$ size, $k < n/2$ is a parameter, and the $\AC^0[2]$ part has depth 6. More precisely, $C'$ is an OR of $2^k$ copies of the depth-5 circuit $C$, and each copy has its first $k$ inputs assigned to a distinct string from $\{0,1\}^k$. Clearly, this circuit $C'$ is satisfiable if and only if $C$ is satisfiable.

Now we wish to evaluate $C'$ on all $2^{n-k}$ inputs, efficiently. Rather than applying Beigel-Tarui at this point, as in Lemma~\ref{accthr2symsym}, we instead apply the probabilistic polynomials of Smolensky~\cite{Smolensky87} to convert $C'$ into a $\SYM \circ \SYM$ circuit $C''$. In particular, we use a slight modification of Smolensky's construction, as described by Kopparty and Srinivasan~\cite{KS12}.

\begin{theorem}[\cite{Smolensky87,KS12}] \label{probpoly} For every $\AC^0$ circuit $C$ of depth $d$, size $s$, and $n$ inputs, and $\eps > 0$, there is a distribution of $n$-variate polynomials ${\cal D}_C$ over $\F_2$ with the following properties. Each $p$ with nonzero support in ${\cal D}_C$ has degree at most $(4\log s)^{d-1}\cdot (\log 1/\eps)$, a polynomial $p$ can be sampled from ${\cal D}_C$ in $n^{O(\log s)^{d-1}(\log 1/\eps)}$ time, and for every $x \in \{0,1\}^n$, $\Pr_{p \thicksim {\cal D}_C}[p(x) = C(x)] \geq 1-\eps$. 
\end{theorem}

We apply Theorem~\ref{probpoly} as follows. Recall that $C'$ is an OR of some $\AC^0[2] \circ \SYM$ circuits $C_1,\ldots,C_{2^k}$, each with (the same) $n-k$ inputs. Moreover, the top $\AC^0[2]$ part of each $C_i$ has depth 5, and each $C_i$ takes $\poly(s,M)$ inputs (coming from the outputs of $\SYM$ gates). For every $i$, we take the top $\AC^0$ part of $C_i$, and invoke Theorem~\ref{probpoly} with $\eps = 1/(10 \cdot 2^k)$ to sample $p_i \thicksim {\cal D}_{C_i}$ of degree at most $O(k(\log s)^{4})$ and  at most $\poly(s,M)^{O(k \cdot (\log s)^{4})}$ monomials. We replace the $\AC^0$ part of $C_i$ with the XOR of ANDs circuit $p_i$. Now the circuit $C'$ is an OR of $2^k$ XOR of AND of SYM circuits; call them $C''_1,\ldots,C''_{2^k}$. For every input $x \in \{0,1\}^{n-k}$, the $\SYM$ gates of $C'$ produce a single $\poly(s,M)$-bit length input $y$. Taking the union bound over all $2^k$ subcircuits, every $C''_1,\ldots,C''_{2^k}$ outputs the same values as $C_1,\ldots,C_{2^k}$ on $x$, with probability at least $1-1/10$. 

Now we randomly convert the topmost OR in $C'$ to an XOR, with the usual Razborov-Smolensky subsum trick: we pick $r_{1,1},r_{2,1},r_{1,2},r_{2,2},\ldots,r_{1,2^k},r_{2,2^k} \in \{0,1\}$ uniformly at random, and replace $C = \OR(C''_1,\ldots,C''_{2^k})$ with \begin{eqnarray*}
C''(x_1,\ldots,x_{n-k}) &:=& \left(\sum_{i=1}^{2^k} r_{1,i}\cdot C''_i(x_1,\ldots,x_{n-k})  \bmod 2\right)\vee \left(\sum_{i=1}^{2^k} r_{2,i}\cdot C''_i(x_1,\ldots,x_{n-k}) \bmod 2\right)\\
& = & \sum_{i=1}^{2^k} r_{1,i}\cdot C''_i(x_1,\ldots,x_{n-k})  + \sum_{i=1}^{2^k} r_{2,i}\cdot C''_i(x_1,\ldots,x_{n-k})\\
& &  + \left(\sum_{i=1}^{2^k} r_{1,i}\cdot C''_i(x_1,\ldots,x_{n-k})\right)\cdot \left(\sum_{i=1}^{2^k} r_{2,i}\cdot C''_i(x_1,\ldots,x_{n-k})\right) \bmod 2,
\end{eqnarray*}
which means that $C''$ equals
\[\sum_{i=1}^{2^k} r_{1,i}\cdot C''_i(x_1,\ldots,x_{n-k})  + \sum_{i=1}^{2^k} r_{2,i}\cdot C''_i(x_1,\ldots,x_{n-k}) + \sum_{i,j=1}^{2^k} r_{1,i}\cdot r_{2,j} \cdot C''_i(x_1,\ldots,x_{n-k})\cdot C''_i(x_1,\ldots,x_{n-k}) \bmod 2.\]
Now for every $x \in \{0,1\}^{n-k}$, 
\begin{eqnarray*}
\lefteqn{\Pr_{p_i \thicksim {\cal D}, r_{i,j} \in \{0,1\}}[C''(x) \neq C'(x)]}\\
& \leq & \Pr_{p_1,\ldots,p_{2^k} \thicksim {\cal D}_{C_i}}[\text{$\exists$ $i$}, C''_i(x) \neq C_i(x)] + \Pr_{r_{i,j} \in \{0,1\}}[\OR(C''_1(x),\ldots,C''_{2^k}(x)) = C'(x) ~|~ \text{$\forall$ $i$}, C''_i(x) = C_i(x)]\\
& \leq & 1/10 + 1/4 \leq 1/3.\end{eqnarray*}

That is, for every input $x \in \{0,1\}^{n-k}$, the probability that $C'(x)=C''(x)$ will be greater than $2/3$. 

Since each polynomial $p_i$ has degree at most $O(k \cdot (\log s)^4)$, the AND gates representing the monomials of $p_i$ have $t \leq O(k \cdot (\log s)^4)$ fan-in. Applying another part of Lemma~\ref{accthr2symsym}, the $\AND_t \circ \SYM$ subcircuits of $C''$ with $\poly(s,M)$ wires can be replaced by a single $\SYM$ gate with $\poly(s,M)^{O(t)}$ input wires. This results in an $\XOR \circ \SYM$ circuit $C''$ of $\poly(s,M)^{O(k \cdot (\log s)^4)}$ total wires; this is also a $\SYM\circ\SYM$ circuit. 

Let $\eps > 0$ be a parameter, and set $k := \max\{1,\frac{\eps n}{(\log M)(\log s)^{5}}\}$. (Note that if $k=1$, the statement of Theorem~\ref{IP} is trivially true.) Following the proof of Theorem~\ref{evalaccthr}, we can apply fast rectangular matrix multiplication to evaluate $C''$ on all $2^{n-k}$ inputs. For sufficiently small $\eps > 0$, the matrix multiplication runs in time 
\[2^{n-k}\cdot \poly(O(k \cdot (\log s)^{4}),\log M, n-k) + \poly(s,M)^{O(k \cdot (\log s)^{4})} \leq 2^{n-\Omega\left(\frac{n}{(\log M)(\log s)^{5}}\right)}\cdot \poly(s,M,n).\] The output of this procedure is a $2^{n-k}$-bit string which, for every $x \in \{0,1\}^{n-k}$, contains the correct output $C'(x)$ with probability at least $2/3$. 

Suppose we repeat the above randomized procedure for $n^2$ times: that is, for $n^2$ times, we independently sample $2^k$ polynomials $p_i$ for each $C_i$ and sample $r_{i,j} \in \{0,1\}$, constructing $n^2$ different circuits $C''_1,\ldots,C''_{n^2}$ from $C'$. Then, standard tail bound arguments show that the majority value output by $C''_1(x),\ldots,C''_{n^2}(x)$ equals $C'(x)$ for every $x \in \{0,1\}^{n-k}$, with high probability. If some assignment $x^{\star}$ has majority value $1$, we conclude that the integer program is \emph{feasible}; otherwise, we output \emph{infeasible}.
\end{proof}

\subsection{Non-uniform $\ACC\circ \THR$ lower bounds}

We now turn to the main application of the evaluation algorithm:

\begin{reminder}{Thm~\ref{nexpaccthr}} $\NEXP$ does not have non-uniform $\ACC \circ \THR$ circuits of quasi-polynomial size.
\end{reminder}

To set the context, let us discuss the prior connection between known circuit satisfiability algorithms and circuit lower bounds. 

\begin{definition} Let ${\cal C}$ be a circuit class. ${\cal C}$ is said to be \emph{typical} if, given any circuit $D$ from one of the classes ${\cal C} \circ {\cal C}$, $\AND \circ {\cal C}$, $\OR \circ {\cal C}$, ${\sf NOT} \circ {\cal C}$, an equivalent $D' \in {\cal C}$ can be produced in $\poly(\text{size}(D))$ time.
\end{definition}

That is, ${\cal C}$ is typical if it is \emph{efficiently closed under composition, unbounded fan-in AND, OR, and negations}. Most well-studied circuit classes have this property. 

From prior work, we know there are connections between the existence of good SAT algorithms for typical circuit classes, and lower bounds against those classes:

\begin{theorem}[\cite{Williams11}]\label{typicalold} Let ${\cal C}$ be typical. Suppose for every $c \geq 1$, there is an $\eps > 0$ and an an algorithm for satisfiability of ${\cal C}$ circuits running in time $O(2^{n-n^{\eps}})$ on circuits with $n$ inputs and $n^{\log^c n}$ size. Then $\NEXP$ does not have quasi-polynomial size ${\cal C}$ circuits.
\end{theorem}

For example, the proof that $\NEXP \not\subset \ACC$ follows from giving a faster-than-exhaustive-search ACC satisfiability algorithm, noting that $\ACC$ is typical, and applying Theorem~\ref{typicalold}. 

This theorem cannot be directly applied to a class such as $\ACC\circ\THR$, because it is not known whether $\ACC\circ\THR\circ\ACC\circ\THR$ can be efficiently simulated with $\ACC\circ\THR$. However, by modifying the argument of Theorem~\ref{typicalold} and using an algorithm for \emph{counting} SAT assignments, we can extend the theorem to circuits with a very weak closure property.\footnote{See also ~\cite{JMV13,Oliveira13} which consider other (stronger) closure properties.}

\begin{definition} Let ${\cal C}$ be a circuit class. We say ${\cal C}$ is \emph{weakly closed under AND} if, given the AND of two circuits of ${\cal C}$, an equivalent circuit in ${\cal C}$ can be produced in polynomial time. 
\end{definition}

Weak closure under AND is satisfied by strictly more circuit classes than the property of being typical. To give an example, any class of the form $\SYM \circ \cdots$ is weakly closed under AND, because an AND of $t$ $\SYM$ gates with $s$ wires can be collapsed into a single symmetric gate with $O(s^t)$ wires (as seen in the proof of Lemma~\ref{accthr2symsym}). However, classes like $\SYM \circ \SYM$ are \emph{not} known to be efficiently closed under composition or unbounded-fan in AND/OR, hence Theorem~\ref{typicalold} does not apply to such classes. We prove:

\begin{theorem}\label{weakAND} Let ${\cal C}$ be weakly closed under AND. Suppose for every $c \geq 1$, there is an $\eps > 0$ and an algorithm for counting the satisfying assignments of ${\cal C}$ circuits in time $O(2^{n-n^{\eps}})$ on circuits with $n$ inputs and $n^{\log^c n}$ size. Then $\NEXP$ does not have quasi-polynomial size ${\cal C}$ circuits.
\end{theorem}

Note that Theorem~\ref{nexpaccthr} (the $\ACC\circ\THR$ lower bound) follows immediately from Theorem~\ref{weakAND} and the counting algorithm of Theorem~\ref{COUNTaccthr}. It is our hope that Theorem~\ref{weakAND} may be applicable in the future to depth-two classes, such as $\SYM\circ\SYM$ and depth-two \emph{exact} threshold circuits~\cite{HP10}: an nontrivial counting SAT algorithm for one of these classes would entail new lower bounds.

\begin{proofof}{Theorem~\ref{weakAND}} (Sketch) Let us start with ${\cal C}$ as typical. We survey what is needed to conclude ${\cal C}$ lower bounds in the proof of Theorem~\ref{typicalold}, and show that the new hypothesis supplies these needs.

The idea is to show that $\NEXP \subset {\cal C}$ and the hypothesis implies every $L \in {\sf NTIME}[2^n]$ can be simulated in nondeterministic $2^n/n$ time, contradicting the  nondeterminstic time hierarchy~\cite{Zak83}. In particular, the assumptions imply that the $\NEXP$-complete problem {\sc Succinct 3SAT} on circuits of AND/OR/NOT with fan-in two, $n$ inputs, and $\poly(n)$ size can be nondeterministically solved in $O(2^{n-n^{\eps}})$ time, which is also provably false~\cite{Williams11survey}. Recall that 	{\sc Succinct 3SAT} is the problem: \emph{given an AND/OR/NOT circuit $C$ of fan-in two, does the truth table of $C$ encode a satisfiable 3-CNF formula?} That is, {\sc Succinct 3SAT} is a ``compressed'' version of the 3SAT problem. 

Suppose we are given an (arbitrary) circuit $C$ of size $s$ and wish to determine if it is a yes-instance of {\sc Succinct 3SAT}. Assuming $\NEXP$ has quasipolynomial-size circuits, it is proved that for every $C$ encoding a satisfiable 3-CNF $F$, there is a quasipolynomial-size circuit $D$ which succinctly encodes a satisfying assignment for $F$: for all $i$, $D(i)$ outputs the value of variable $x_i$ in the satisfying assignment. Our ``fast'' nondeterministic algorithm for {\sc Succinct 3SAT} guesses this circuit $D$, and uses it to construct a circuit $E$ with $n$ inputs and $n^{\log^c n}$ size for some $c$, which is unsatisfiable if and only if $D$ encodes a satisfying assignment to the formula $F$ encoded by $C$. 

Assuming $\NEXP$ has quasipolynomial-size ${\cal C}$ circuits and that there is an $O(2^{n-n^{\eps}})$ time algorithm for ${\cal C}$ satisfiability, it is proved that there is a nondeterministic algorithm $A$ running in $2^{n-\Omega(n^{\eps})}$ time which, given an AND/OR/NOT of fan-in two circuit $E$ of size $s$ and $n$ inputs, outputs an equivalent $E'$ of $s^{\log^c s}$ size from the class ${\cal C}$ on at least one nondeterministic branch (and prints \emph{no} on other branches). Running this algorithm $A$, obtaining $E'$, then running the ${\cal C}$ satisfiability algorithm on $E'$, we nondeterministically determine that $C$ is a yes-instance of {\sc Succinct-3SAT} in $2^{n-\Omega(n^{\eps})}$ time.

Now assume ${\cal C}$ is weakly closed under AND. The point where closure properties are relevant is precisely in the argument that the nondeterministic algorithm $A$ exists. In fact, if our hypothesis and the assumption that $\NEXP$ has quasipolynomial-size ${\cal C}$ circuits implies such an algorithm, it can be observed that the rest of the proof carries over without modification. We now construct such an algorithm $A$.

The algorithm $A$ starts by guessing a ${\cal C}$ circuit $E''$ of $n^{\log^c n}$ size which takes as input a pair $(x,g) \in \{0,1\}^n\times \{0,1\}^{\log(\text{size}(E))}$, and outputs $1$ if and only if the gate $g$ in $E$ outputs $1$ when $E$ is given the input $x$. (Such an $E''$ exists, assuming $\P$ has quasi-polynomial size ${\cal C}$ circuits.) 

Now we need to verify that for every gate $g$ indexed by $1,2,\ldots,\text{size}(E)$, $E''(x,g)$ outputs what gate $g$ of $E(x)$ outputs, on all $x$. Each gate $g$ is either an input, an AND of two previous gates $g_1$ and $g_2$, an OR of two previous gates $g_1$ and $g_2$, or a NOT of a previous gate $g_1$. 

To aid this verification, we show how to efficiently check for arbitrary ${\cal C}$ circuits $G$ and $H$ whether $G(x)=H(x)$ for all inputs $x$, using an algorithm for counting SAT assignments. Let $\#SAT(C)$ be the number of satisfying assignments to a circuit $C$. 
Observe that $G(x)=H(x)$ for all $x$ if and only if $\#SAT(G)=\#SAT(H)=\#SAT(G \wedge H)$. (Note the third quantity can be efficiently computed, assuming ${\cal C}$ is weakly closed under AND.) Moreover,  $G(x) \neq H(x)$ for all $x$ if and only if $\#SAT(G) + \#SAT(H)=2^n$ and $\#SAT(G \wedge H)=0$. Therefore, by counting SAT assignments, we have algorithms checking whether $G$ is equivalent to $H$, and whether $G$ is equivalent to the negation of $H$, both running in time $O(2^{n-n^{\eps}})$.

We claim that the verification problem for $E''$ can be reduced to a number of calls to the above kinds of checks. First, nondeterministically guess a circuit $E''_{not}$, intended to satisfy $E''_{not}(x,g) = \neg E''(x,g)$ for all $x$ and $g$. Verifying this condition can be done by counting SAT assignments, as described above. 

Checking $E''$ is correct on the input gates of $E$ means that for all $i=1,\ldots,n$, $E''(x_1,\ldots,x_n,i) = x_i$. Both $E''(x_1,\ldots,x_n,i)$ and $I(x_1,\ldots,x_n) = x_i$ are ${\cal C}$ circuits, hence their equivalence can be verified by $\#$SAT calls. Checking a NOT gate $g$ of $E$ with input gate $g_1$ is equivalent to checking that $E''_{not}(x,g_1) = E''(x,g)$ on all $x$. Checking an AND gate $g$ of two previous gates $g_1$ and $g_2$ amounts to checking that $E''(x,g) = E''(x,g_1)\wedge E''(x,g_2)$ on all $x$. To do this, compute $G_{and}(x) := E''(x,g_1)\wedge E''(x,g_2)$ (assuming ${\cal C}$ is weakly closed under AND), then check $G_{and}(x) = E''(x,g)$ for all $x$. Finally, for an OR gate $g$ with inputs $g_1$ and $g_2$, we want to check that $E''(x,g) = E''(x,g_1)\vee E''(x,g_2)$ on all $x$. This is equivalent to $\neg E''(x,g) = ((\neg E''(x,g_1))\wedge (\neg E''(x,g_2)))$ for all $x$. This can be checked by forming $G_{or}(x) := E''_{not}(x,g_1)\wedge E''_{not}(x,g_2)$, then checking that $G_{or}(x) = E''_{not}(x,g)$ for all $x$.

On a circuit $E$ with $s \leq n^{\log^c n}$ gates, the above procedure runs in $O(2^{n-n^{\eps}}\cdot s) \leq 2^{n-\Omega(n^{\eps})}$ time. When it concludes, we know that for all gates $g$ and all $x$ that $E''(x,g)$ outputs the correct value. The circuit $E'(x)$ output by $A$ simply evaluates $E''(x,g^{\star})$, where $g^{\star}$ is the output gate of $E$.
\end{proofof}

\section{Fast evaluation of depth-two threshold circuits}

Finally, we show a strong sense in which depth-two threshold circuits are \emph{weak}, by giving a fast algorithm for evaluating such circuit on many assignments in batch. The general theorem is:

\begin{theorem}\label{thrthreval} Given a depth-two linear threshold circuit $C$ with $2k$ inputs and at most $n^{1/12}$ gates with weights on the bottom layer of absolute value at most $W_b$, weights on the output gate of absolute value at most $W_o$, and given two sets $A,B \subseteq \{0,1\}^k$ where $|A|=|B|=n$, we can evaluate $C$ on all $n^2$ points in $A \times B$ using $n^2 \cdot \poly(\log W_o,\log n) + n^{1+1/12}\cdot \poly(\log n, \log W_b)$ time.
\end{theorem}

The following is immediate from Theorem~\ref{thrthreval}:

\begin{reminder}{Theorem~\ref{depth2eval}} Let $k > 1$. Given a depth-two $2^{n/24}$-size linear threshold circuit $C$ with integer weights in $[-2^{n^k},-2^{n^k}]$, we can evaluate $C$ on all $2^n$ input assignments in $2^n \cdot \poly(n^k)$ time.
\end{reminder}

While the proof of Theorem~\ref{thrthreval} also ultimately depends on Coppersmith's rectangular matrix multiplication, the rest of the algorithm is rather different from the evaluation algorithm of Theorem~\ref{evalaccthr}. 

\begin{proofof}{Theorem~\ref{thrthreval}} We reduce the evaluation task to a special kind of matrix multiplication, then combine Coppersmith's matrix multiplication with a mild brute force to expedite the matrix multiply. 

Define $\text{LEQ} : \Z \times \Z \rightarrow \{0,1\}$ to output $1$ on $(a,b)$ if and only if $a \leq b$. Given a vector $w = (w_1,\ldots,w_d) \in \Z^d$, and given two matrices $M$ and $N$ which are $n \times d$ and $d \times n$, define their \emph{$w$-weighted threshold product} to be $(M \circledast_w N)[i,j] := \sum_{k=1}^d w_k \cdot \text{LEQ}(M[i,k],N[k,j])$.

We shall show that the $w$-weighted threshold product of an $n \times n^{1/12}$ matrix and an $n^{1/12}\times n$ matrix can be computed in essentially $n^2 \cdot \poly(\log n)$ time (with some additional but negligible overhead in terms of the weights). Let us postpone this algorithm for the moment, and first show how to embed the evaluation problem into the weighted threshold product.

Let $C$ be a depth-two circuit of size $s$, with the $2k$ input variables $x_1,\ldots,x_k,y_1,\ldots,y_k$. Let $w_1,\ldots,w_s$ be the weights of the top threshold gate of $C$, and let $\ell_1,t_1, \ldots,\ell_s, t_s$ be the corresponding linear forms and threshold values from the bottom layer of threshold gates: that is, the output of $\text{LEQ}(t_i, \ell_i)$ is multipled by $w_i$ in the output gate. Without loss of generality, we may assume that all weights $w_i$ are multiplied by the output of some threshold gate at the bottom layer (there are at most $n$ wires from the input directly to the output gate, and they can be replaced by $O(n)$ dummy gates at the bottom layer with wires to the output gate). Let $A = \{A_1,\ldots,A_n\} \subseteq \{0,1\}^k$ and $B = \{B_1,\ldots,B_n\}\subseteq\{0,1\}^k$. 

We partition each linear form $\ell_j$ on the bottom layer into two sums $\ell_j^{(x)}$ and $\ell_j^{(y)}$, such that $\ell_j^{(x)}$ involves only input variables $x_1,\ldots,x_k$, $\ell_j^{(y)}$ involves only $y_1,\ldots,y_k$, and $\ell_j^{(x)} + \ell_j^{(y)} = \ell_j$. Let $A_i(\ell_j^{(x)})$ and $B_j(\ell_j^{(y)})$ denote the value of the linear form $\ell_j^{(x)}$ (respectively, $\ell_j^{(y)}$) evaluated on assignment $A_i$ (respectively, $B_j$).

Define the matrix $M$ with rows indexed by elements of $A$, and columns indexed by the bottom layer gates $1,\ldots,s$. Set $M[i,k]$ to the value $t_k - A_i(\ell_k^{(x)})$. The matrix $N$ has rows indexed by the bottom layer gates $1,\ldots,s$, and columns indexed by elements of $B$. Set $N[k,j]$ to the value $B_j(\ell_k^{(y)})$.

Now consider the $w$-weighted threshold product $M \circledast_w N$, where $w$ is the same as above. The $i,j$ entry of this product equals \[\sum_{k=1}^s w_k \cdot \text{LEQ}\left(t_k - A(\ell_k^{(x)}), B_j(\ell_k^{(y)})\right) = \sum_{k=1}^s w_k \cdot \text{LEQ}\left(t_k, A_i(\ell_k^{(x)})+B_j(\ell_k^{(y)})\right).\] 
This is precisely the value of the linear form in the output gate of $C$, when $x_1,\ldots,x_k$ are given the assignment $A_i$ and $y_1,\ldots,y_k$ are assigned $B_j$. The truth table of $C$ on $A \times B$ can be recovered by simply checking which entries in $(M \circledast_w N)$ exceed the output gate's threshold.

Next, we shall show how to compute a weighted threshold matrix product efficiently. Let $\delta$ be a parameter, and let $M$ and $N$ be $n \times n^{\delta}$ and $n^{\delta} \times n$ matrices, respectively. The first step is to reduce the weights significantly. For all $k=1,\ldots,n^{\delta}$, let $S_k$ be a list of all entries in the $k$th column of $M$, plus the $k$th row of $N$. Sort $S_k$, obtaining a ranking of $2n$ items, and replace each entry in the $k$th column of $M$ and the $k$th row of $N$ by their rank in the sorted list $S_k$. This step reduces the domains of $M$ and $N$ to $\{1,\ldots,2n\}$, and the $w$-weighted threshold matrix product remains the same: all inequalities $M[i,k] \leq N[k,j]$ are preserved. Note this step takes $n^{1+\delta}\cdot\poly(\log n,\log W_b)$ time.

In order to reduce to matrix multiplication, we perform two strategies with different advantages. (The reduction is inspired by work of Matousek~\cite{Matousek91} on computing dominances in high dimensions.) Let $s \in \{1,\ldots,n\}$ be a parameter. Partition each sorted list $S_k$ into $t =\lceil n/s \rceil$ contiguous buckets $T_1,\ldots,T_t$, where each bucket $T_i$ contains at most $s$ entries. (For all $i < j$, the largest entry in $T_i$ is at most the smallest entry in $T_j$.)

Start with an $n \times n$ output matrix $P$ that is all zeroes. For every $(i,k) \in [n] \times [n^{\delta}]$, look up the bucket $T_{\ell}$ containing $M[i,k]$ in the sorted list $S_k$. For all $N[k,j]$ contained in $T_{\ell}$ such that $M[i,k] \leq N[k,j]$, add the weight $w_k$ to the entry $P[i,j]$. This loop adds to $P$ all terms $w_k \cdot \text{LEQ}(M[i,k],N[k,j])$ such that $M[i,k]$ and $N[k,j]$ appear in the same bucket of $S_k$. Observe that this step takes $\tilde{O}(n \cdot n^{\delta} \cdot s)$ time. 

To handle the $(M[i,k],N[k,j])$ pairs that do not appear in the same bucket, we use matrix multiplication. For each $(i,k) \in [n] \times [n^{\delta}]$, replace the entry $M[i,k]$ with a row vector $v_{i,k} \in \{0,w_k\}^t$, such that $v_{i,k}[\ell] := w_k$ if and only if $M[i,k]$ is in bucket $T_{\ell}$ of $S_k$. That is, $v_{i,k}$ has $w_k$ in exactly one entry, and zeroes elsewhere. This forms a matrix $M'$ of dimensions $n \times (n^{\delta} \cdot t)$. For $(k,j) \in [n^{\delta}]\times [n]$, replace each entry $N[k,j]$ with a column vector $u_{k,j} \in \{0,1\}^t$, such that $v_{i,k}[\ell'] := 1$ if and only if $N[k,j]$ is in bucket $T_{\ell}$ of $S_k$ and $\ell > \ell'$. This forms a matrix $N'$ of dimensions $(n^{\delta} \cdot t) \times n$. The matrix product $M' \cdot N'$ over the integers computes a sum of inner products \[(M' \cdot N')[i,j] = \sum_{n^{\delta}} \langle v_{i,k},u_{k,j}\rangle.\] If $M[i,k]>N[k,j]$, or $M[i,k]$ and $N[k,j]$ are in the same bucket of $S_k$, then $\langle v_{i,k},u_{k,j}\rangle = 0$. If $M[i,k]\leq N[k,j]$ but $N[k,j]$ and $M[i,k]$ are in different buckets of $S_k$ then $\langle v_{i,k},u_{k,j}\rangle = w_k$. 

Letting $P := P + (M'\cdot N')$, this procedure adds to $P$ all terms $w_k \cdot \text{LEQ}(M[i,k],N[k,j])$ such that $M[i,k]$ and $N[k,j]$ appear in different buckets of $S_k$. Therefore $P[i,j]$ contains the value of the linear form for the output gate of $C$, under variable assignment $(A_i,B_j)$, for all $i,j$.

The above algorithm runs in time $O(n\cdot n^{\delta}\cdot s\log W_o + MM(n,n^{1+\delta}/s,n)\cdot \poly(\log W_o))$, where $MM(a,b,c)$ is the running time for multiplying $a \times b$ and $b \times c$ matrices. If we set $n^{1+\delta}/s = n^{0.172}$, then Coppersmith's algorithm (Lemma~\ref{coppersmith}) can be applied to the second term of the running time, implementing it in $n^2 \cdot \poly(\log n)$ time. Under this setting, $s = n^{\delta} \cdot n^{0.828}$ and the first term of the running time is $n^{1+2\delta+0.828}$. Setting $\delta = 0.086 > 1/12$, the first term becomes $n^2$ (note that $s = n^{.914}$). 
\end{proofof}

It is easy to see that, since the above algorithm actually evalutes the linear form at the output gate of a depth-two threshold circuit, we can also efficiently evaluate large $\SYM \circ \THR$ circuits as well. 

\paragraph{Acknowledgements.} I thank Igor Carboni Olivera for sending a preliminary version of his survey, which helped the ideas in the proof of Theorem~\ref{weakAND} to congeal. I also thank Rahul Santhanam for helpful comments on an earlier draft.

\bibliographystyle{alpha}
\bibliography{papers,apsp}

\appendix

\section{Appendix: An exposition of Coppersmith's algorithm}
\label{coppersmith-appendix}

In 1982, Don Coppersmith proved that the rank (that is, the number of essential multiplications) of $N \times N^{0.172}$ and $N^{0.172} \times N$ matrix multiplication is at most $O(N \log^2 N)$. Prior work has observed that his algorithm can also be used to show that the total number of arithmetic operations for the same matrix multiply is $N \cdot \poly(\log N)$. However, the implication is not immediate, and uses specific properties of Coppersmith's algorithm. Because this result is so essential to this work and a recent algorithm for all-pairs shortest paths~\cite{Williams13APSP}, we give here a self-contained exposition.

\begin{theorem}[Coppersmith~\cite{Coppersmith82}]\label{rank} For all sufficiently large $N$, the rank of $N \times N^{.172} \times  \times N$ matrix multiplication is at most $O(N^2 \log^2 N)$. \end{theorem}

We wish to derive the following consequence of Coppersmith's construction, which has been mentioned in the literature before~\cite{Seroussi,ApplebaumCPS09,Williams11}:

\begin{reminder}{Lemma~\ref{coppersmith}} For all sufficiently large $N$, and $\alpha \leq .172$, multiplication of an $N \times N^{\alpha}$ matrix with an $N^{\alpha} \times N$ matrix can be done in $N^2 \cdot \poly(\log N)$ arithmetic operations, over any field with $O(2^{\poly(\log N)})$ elements.\end{reminder}

For brevity, we will use the notation ``$\ell \times m \times n$ matrix multiply'' to refer to the multiplication of $\ell \times m$ and $m \times n$ matrices (hence the above gives an algorithm for $N \times N^{\alpha} \times N$ matrix multiply).

Note Lemma~\ref{coppersmith} has been ``improved'' in the sense that the upper bound on $\alpha$ has been increased mildly over the years~\cite{Coppersmith97,HP98,Ke08,LeGall12}. However, these later developments only run in $N^{2+o(1)}$ time, not $N^2 \cdot \poly(\log N)$ time (which we require). Our exposition will expand on the informal description given in recent work~\cite{Williams11}.

First, observe that the implication from Theorem~\ref{rank} to Lemma~\ref{coppersmith} is not immediate. For example, it could be that Coppersmith's algorithm is non-uniform, making it difficult to apply. As far as we know, one cannot simply take ``constant size'' arithmetic circuits implementing the algorithm of Theorem~\ref{rank} and recursively apply them. In that case, the $\poly(\log N)$ factor in the running time would then become $N^{\eps}$ for some constant $\eps > 0$ (depending on the size of the constant-size circuit). To keep the overhead polylogarithmic, we have to unpack the algorithm and analyze it directly.

\subsection{A short preliminary}

Coppersmith's algorithm builds on many other tools from prior matrix multiplication algorithms, many of which can be found in the highly readable book of Pan~\cite{Pan84}. Here we will give a very brief tutorial of some of the aspects.

\paragraph{Bilinear algorithms and trilinear forms.} Essentially all methods for matrix multiplication are bilinear (and if not, they can be converted into such algorithms), meaning that they can be expressed in the so-called trilinear form
\begin{equation}\label{trilinear} \sum_{ijk} A_{ik}B_{kj}C_{ji} + p(x) = \sum_{\ell=1}^5 (\sum_{ij} \alpha_{ij} A_{ij})\cdot (\sum_{ij} \beta_{ij} B_{ij})\cdot(\sum_{ij} \gamma_{ij} C_{ij})\end{equation}
where $\alpha_{ij}$, $\beta_{ij}$, and $\gamma_{ij}$ are constant-degree polynomials in $x$ over the field, and $p(x)$ is a polynomial with constant coefficient $0$. Such an algorithm can be converted into one with no polynomials and minimal extra overhead (as described in Coppersmith's paper). Typically one thinks of $A_{ik}$ and $B_{kj}$ as entries in the input matrices, and $C_{ji}$ as indeterminates, so the LHS of~\eqref{trilinear} corresponds to a polynomial whose $C_{ji}$ coefficient is the $ij$ entry of the matrix product. Note the {\em transpose} of the third matrix $C$ corresponds to the final matrix product.

To give an explicit example, we assume the reader is familiar with Strassen's famous method for $2 \times 2 \times 2$ matrix multiply. Strassen's algorithm can be expressed in the form of~\eqref{trilinear} as follows:
\begin{eqnarray}\label{strassen}
\sum_{i,j,k=0,1} A_{ik}B_{kj}C_{ji} &=& (A_{00} + A_{11})(B_{00} + B_{11})(C_{00} + C_{11})\\ \nonumber
& & + (A_{10} + A_{11})B_{00}(C_{01}-C_{11}) + A_{00}(B_{01} - B_{11})(C_{10} + C_{11})\\ \nonumber
& & + (A_{10} - A_{00})(B_{00} + B_{01})C_{11} + (A_{00} + A_{01})B_{11}(C_{10} - C_{00}) \\ \nonumber
& & + A_{11}(B_{10} - B_{00})(C_{00} + C_{01}) + (A_{01} - A_{11})(B_{10}+B_{11})C_{00}. \nonumber\end{eqnarray}
The LHS of~\eqref{trilinear} and~\eqref{strassen} represents the trace of the product of three matrices $A$, $B$, and $C$ (where the $ij$ entry of matrix $X$ is $X_{ij}$). It is well known that every bilinear algorithm naturally expresses multiple algorithms through this trace representation. Since \[tr(ABC) = tr(BCA) = tr(CAB)=tr((ABC)^T)=tr((BCA)^T)=tr((CAB)^T),\] if we think of $A$ as a symbolic matrix and consider~\eqref{trilinear}, we obtain a new algorithm for computing a matrix $A$ when given $B$ and $C$. Similarly, we get an algorithm for computing a $B$ when given $A$ and $C$, and analogous statements hold for computing $A^T$, $B^T$, and $C^T$. So the aforementioned algorithm for multiplying a sparse $2 \times 3$ and sparse $3 \times 2$ yields several other algorithms.

\paragraph{Sch\"{o}nhage's decomposition paradigm.} Coppersmith's algorithm follows a specific paradigm introduced by Sch\"{o}nhage~\cite{Schoenhage81} which reduces arbitrary matrix products to slightly larger matrix products with ``structured nonzeroes.'' The general paradigm has the following form. Suppose we wish to multiply two matrices $A''$ and $B''$.
\begin{compactenum}
\item  First we {\em preprocess} $A''$ and $B''$ in some efficient way, decomposing $A''$ and $B''$ into structured matrices $A,A',B,B'$ so that $A'' \cdot B'' = A' \cdot A \cdot B \cdot B'$. (Note, the dimensions of $A' \cdot A$ may differ from $A''$, and similarly for $B' \cdot B$ and $B''$.) The matrices $A$ and $B$ are sparse ``partial'' matrices directly based on $A''$ and $B''$, but they have larger dimensions, and only contain nonzeroes in certain structured parts. The matrices $A'$ and $B'$ are very simple and explicit matrices of scalar constants, chosen independently of $A''$ and $B''$. (In particular, $A'$ and $B'$ are Vandermonde-style matrices.)

\item Next, we apply a specialized constant-sized matrix multiplication algorithm in a recursive manner, to multiply the structured $A$ and $B$ essentially optimally. Recall that Strassen's famous matrix multiplication algorithm has an analogous form: it starts with a seven-multiplication product for $2 \times 2 \times 2$ matrix multiplication, and recursively applies this to obtain a general algorithm for $2^M \times 2^M \times 2^M$ matrix multiplication. Here, we will use an \emph{optimal} algorithm for multiplying constant-sized matrices with zeroes in some of the entries; when this algorithm is recursively applied, it can multiply sparse $A$ and $B$ with nonzeroes in certain structured locations.

\item Finally, we {\em postprocess} the resulting product $C$ to obtain our desired product $A'' \cdot B''$, by computing $A' \cdot C \cdot B'$. Using the simple structure of $A'$ and $B'$, the matrix products $D := A' \cdot C$ and $D \cdot B'$ can be performed very efficiently. Our aim is to verify that each step of this process can be efficiently computed, for Coppersmith's full matrix multiplication algorithm.

\end{compactenum}

\subsection{The algorithm}

The construction of Coppersmith begins by taking input matrices $A''$ of dimensions $2^{4M/5} \times {M \choose 4M/5}2^{4M/5}$ and $B''$ of dimensions ${M \choose 4M/5}2^{4M/5} \times 2^{M/5}$ where $M \approx \log N$, and obtains an $O(5^M \poly(M))$ algorithm for their multiplication. Later, he symmetrizes the construction to get an $N \times N \times N^{\alpha}$ matrix multiply. We will give this starting construction and show how standard techniques can be used to obtain an $N \times N^{\alpha} \times N$ matrix multiply from his basic construction.

The multiplication of $A''$ and $B''$ will be derived from an algorithm which computes the product of $2 \times 3$ and $3 \times 2$ matrices with zeroes in some entries. In particular the matrices have the form:
\[\left( \begin{array}{ccc} a_{11} & a_{12} & a_{13} \\ 0 & a_{22} & a_{23} \end{array} \right), \left( \begin{array}{cc} b_{11} & b_{12}\\ b_{21} & 0 \\ b_{31} & 0 \end{array} \right), \] and the algorithm is given by the trilinear form \begin{eqnarray}\label{2332}(a_{11} + x^2 a_{12})(b_{21} + x^2 b_{11})(c_{11}) + (a_{11} +x^2 a_{13}(b_{31})(c_{11} - x c_{21}) + (a_{11} + x^2 a_{22})(b_{21} - x b_{21})(c_{22})\\ \nonumber + (a_{11}+x^2 a_{23})(b_{31} + x b_{12}) (c_{12} + x c_{21}) - (a_{11})(b_{21}+b_{31})(c_{11}+c_{12})\\ \nonumber = x^2(a_{11}b_{11}c_{11} + a_{11}b_{12}c_{21} + a_{12}b_{21}c_{11} + a_{13}b_{31}c_{11} + a_{22}b_{21}c_{12} + a_{23}b_{31}c_{12}) + x^3 \cdot P(a,b,c,x).\end{eqnarray} That is, by performing the five products of the linear forms of $a_{ij}$ and $b_{k\ell}$ on the LHS, and using the $c_{ij}$ to determine how to add and subtract these products to obtain the output $2 \times 2$ matrix, we obtain a polynomial in each matrix entry whose $x^2$ coefficients yield the final matrix product $c_{ij}$.

When the algorithm given by \eqref{2332} is applied recursively to $2^M \times 3^M$ and $3^M \times 2^M$ matrices (analogously to how Strassen's algorithm is applied to do $2^M \times 2^M \times 2^M$ matrix multiply), we obtain an algorithm that can multiply matrices $A$ and $B$ with dimensions $2^M \times 3^M$ and $3^M \times 2^M$, respectively, where $A$ has $O(5^M)$ nonzeroes, $B$ has $O(4^M)$ nonzeroes, and these nonzeroes appear in a highly regular pattern (which can be easily deduced). This recursive application of \eqref{2332} will result in polynomials in $x$ of degree $O(M)$, and additions and multiplications on such polynomials increase the overall time by an $M \cdot \poly(\log M)$ factor. Therefore we can multiply these $A$ and $B$ with structured nonzeroes in $O(5^M \cdot \poly(M))$ field operations.

The decomposition of $A''$ and $B''$ is performed as follows. We choose $A'$ and $B'$ to have dimensions $2^{4M/5} \times 2^{M}$ and $2^M \times 2^{M/5}$, respectively, and such that  all $2^{4M/5} \times 2^{4M/5}$ submatrices of $A'$ and $2^{M/5} \times 2^{M/5}$ submatrices of $B'$ are non-singular. Following Sch\"{o}nhage, we pick $A'$ and $B'$ to be rectangular Vandermonde matrices: the $i,j$ entry of $A'$ is $(\alpha_j)^{i-1}$, where $\alpha_1,\alpha_2,\ldots$ are distinct elements of the field; $B'$ is defined analogously. Such matrices have three major advantages: (1) they can be succinctly described (with $O(2^M)$ field elements), (2) multiplying these matrices with arbitrary vectors can be done extremely efficiently, and (3) inverting an arbitrary square submatrix can be done extremely efficiently. More precisely, $n\times n$ Vandermonde matrices can be multiplied with arbitrary $n$-vectors in $O(n \cdot \poly(\log n))$ operations, and computing the inverse of an $n \times n$ Vandermonde matrix can be done in $O(n \cdot \poly(\log n))$ operations (for references, see~\cite{CKY,Bini-Pan94}). In general, operations on Vandermonde matrices, their transposes, their inverses, and the transposes of inverses can be reduced to fast multipoint computations on univariate polynomials. For example, multiplying an $n\times n$ Vandermonde matrix with a vector is equivalent to evaluating a polynomial (with coefficients given by the vector) on the $n$ elements that comprise the Vandermonde matrix, which takes $O(n \log n)$ operations. This translates to $O(n \cdot \poly(\log n))$ arithmetic operations.

The matrices $A$ and $B$ have dimensions $2^M \times 3^M$ and $3^M \times 2^M$, respectively, where $A$ has  only $O(5^M)$ nonzeroes, $B$ has only $O(4^M)$ nonzeroes, and there is an optimal algorithm for multiplying $2 \times 3$ (with 5 nonzeroes) and $3 \times 2$ matrices (with 4 nonzeroes) that can be recursively applied to multiply $A$ and $B$ optimally, in $O(5^M \cdot \poly(M))$ operations. Matrices $A$ and $B$ are constructed as follows: take any one-to-one mapping between the ${M \choose 4M/5}2^{M/5}$ columns of the input $A''$ and columns of the sparse $A$ with exactly $2^{4M/5}$ nonzeroes. For these columns $q$ of $A$ with $2^{4M/5}$ nonzeroes, we compute the inverse $A_q^{-1}$ of the $2^{4M/5} \times 2^{4M/5}$ minor $A_q$ of $A'$ with rows corresponding to the nonzeroes in the column, and multiply $A_q^{-1}$ with column $q$ (in $2^{4M/5} \cdot \poly(M)$ time). After these columns are processed, the rest of $A$ is zeroed out. Then, there is a one-to-one correspondence between columns of $A''$ and nonzero columns of $A' \cdot A$. Performing a symmetric procedure for $B''$ (with the same mapping on rows instead of columns), we can decompose it into $B$ and $B'$ such that there is a one-to-one correspondence between rows of $B''$ and nonzero rows of $B\cdot B'$. It follows that this decomposition takes only $O({M \choose 4M/5}2^{4M/5} \cdot 2^{4M/5} \cdot \poly(M))$ time. Since $5^M \approx {M \choose 4M/5}4^{4M/5}$ (within $\poly(M)$ factors), this quantity is upper bounded by $5^M \cdot \poly(M)$.

After $A$ and $B$ are constructed, the constant-sized algorithm for $2 \times 3$ and $3 \times 2$ mentioned above can be applied in the usual recursive way to multiply the sparse $A$ and $B$ in $O(5^M \cdot \poly(M))$ operations; call this matrix $Z$. Because $A'$ and $B'$ are Vandermonde, the product $A' \cdot Z \cdot B'$ can be computed in $O(5^M \cdot \poly(M))$ operations. Hence we have an algorithm for multiplying matrices of dimensions $2^{4M/5} \times {M \choose 4M/5}2^{4M/5}$ and ${M \choose 4M/5}2^{4M/5} \times 2^{M/5}$ that is explicit and takes $5^M \cdot \poly(M)$ operations.

Call the above algorithm {\sc Algorithm 1}. Observe {\sc Algorithm 1} also works when the entries of $A''$ and $B''$ are themselves matrices over the field. (The running time will surely increase in proportion to the sizes of the underlying matrices, but the bound on the number of {\em operations on the entries} remains the same.)

Up to this point, we have simulated Coppersmith's construction completely, and have simply highlighted its efficiency. By exploiting the symmetries of matrix multiplication algorithms in a standard way, we can extract more algorithms from the construction. The trace identity tells us that \[tr(ABC) = tr(BCA),\] implying that the expression \eqref{2332} can also be used to partially multiply a $3^M \times 2^M$ matrix $B$ with at most $4^M$ structured nonzeroes and ``full'' $2^M \times 2^M$ matrix $C$ in $5^M \cdot \poly(M)$ operations, obtaining a $3^M \times 2^M$ matrix $A^T$ with at most $5^M$ nonzeroes. In our {\sc Algorithm 1}, we have a decomposition of $A$ and $B$; in terms of the trace, we can derive: \[tr(A'' B'' \cdot C'') = tr(A'A \cdot B B' \cdot C'') = tr(B \cdot B' C'' A' \cdot A).\]

This can be applied to obtain an algorithm for ${M \choose 4M/5}2^{4M/5} \times 2^{M/5} \times 2^{4M/5}$ matrix multiplication, as follows. Given input matrices $B''$ and $C''$ of the respective dimensions, decompose $B''$ into a $3^M \times 2^M$ $B$ with $O(4^M)$ nonzeroes and $2^{M} \times 2^{M/5}$ Vandermonde $B'$, as described above. Letting $A'$ be a Vandermonde $2^{4M/5} \times 2^{M}$ matrix, compute the matrix $C := B' \cdot C'' \cdot A'$ in at most $4^M \cdot \poly(M)$ operations. Noting that $C$ is $2^M \times 2^M$, we can then multiply $B$ and $C$ in $5^M \cdot \poly(M)$ operations. This results in a $3^M \times 2^M$ matrix $A^T$ with at most $5^M$ nonzeroes. The final output $A''$ is obtained by using the one-to-one mapping to extract the appropriate ${M \choose 4M/5}2^{4M/5}$ rows from $A^T$, and multiplying each such row by the appropriate inverse minor of $A'$ (corresponding to the nonzeroes of that row). This takes at most ${M \choose 4M/5}2^{4M/5} \cdot 2^M \cdot \poly(M) \leq 5^M \cdot \poly(M)$ operations. Call this {\sc Algorithm 2}.

From {\sc Algorithm 2} we immediately obtain an algorithm for $2^{4M/5} \times 2^{M/5} \times {M \choose 4M/5}2^{4M/5}$ matrix multiplication as well: given input matrices $(C'')^T$ and $(B'')^T$ of the respective dimensions, simply compute $B'' \cdot C''$ using {\sc Algorithm 2}, and output the transpose of the answer. Call this {\sc Algorithm 3}.

Finally, by ``tensoring'' {\sc Algorithm 2} with {\sc Algorithm 3}, we derive an algorithm for matrix multiplication with dimensions \[{M \choose 4M/5}2^{4M/5}\cdot 2^{4M/5} \times 2^{2M/5} \times {M \choose 4M/5}2^{4M/5}\cdot 2^{4M/5} \geq 5^M/M \times 4^{M/5} \times 5^M/M.\]

That is, we divide the two input matrices of large dimensions into blocks of $2^{4M/5} \times 2^{M/5}$ and $2^{M/5} \times {M \choose 4M/5}2^{4M/5}$ dimenisons, respectively. We execute {\sc Algorithm 2} on the blocks, and call {\sc Algorithm 3} when the product of two blocks is needed.

As both {\sc Algorithm 2} and {\sc Algorithm 3} are explicit and efficient, their ``tensorization'' inherits these properties. {\sc Algorithm 2} uses $5^M \cdot \poly(M)$ operations, and each operation can take up to $5^M \cdot \poly(M)$ time (due to calls to {\sc Algorithm 3}). Therefore, we can perform a $5^M \times 4^{2M/5} \times 5^M$ matrix multiply over fields with $2^{\poly(M)}$ elements, in $5^{2M} \cdot \poly(M)$ time. Setting $n = \log(M)/\log(5)$, the algorithm runs in $n^2 \cdot \poly(\log n)$ time for fields with $2^{\poly(\log n)}$ elements.

\end{document}